\pdfoutput=1 % forces arxiv to use pdflatex for microtype otherwise fails

\documentclass[UKenglish]{article}

\usepackage{scrextend}
\changefontsizes[10pt]{8pt}

\usepackage[sc]{mathpazo} % Use the Palatino font
\usepackage[T1]{fontenc} % Use 8-bit encoding that has 256 glyphs
\usepackage[utf8]{inputenc}
\usepackage{microtype} % Slightly tweak font spacing for aesthetics

\usepackage[a4paper,top=30mm,left=27mm,columnsep=20pt]{geometry} % Document margins
\usepackage{multicol} % Used for the two-column layout of the document
\usepackage[hang, small,labelfont=bf,up,textfont=it,up]{caption} % Custom captions under/above floats in tables or figures
\usepackage{booktabs} % Horizontal rules in tables
\usepackage{float} % Required for tables and figures in the multi-column environment - they need to be placed in specific locations with the [H] (\eg  \begin{table}[H])

\usepackage{lettrine} % The lettrine is the first enlarged letter at the beginning of the text
\usepackage{paralist} % Used for the compactitem environment which makes bullet points with less space between them

\usepackage{abstract} % Allows abstract customization
\usepackage{titlesec} % Allows customization of titles
\titleformat{\section}[block]{\large\scshape\centering}{\thesection.}{1em}{} % Change the look of the section titles
\titleformat{\subsection}[block]{\large}{\thesubsection.}{1em}{} % Change the look of the section titles

\usepackage{fancyhdr} % Headers and footers
\usepackage{verbatim,fancyvrb}
\usepackage{xcolor}
\usepackage{amsmath,amsthm,amsfonts,amssymb,stmaryrd,mathtools,centernot}

\usepackage{tikz}
\usetikzlibrary{calc,arrows,matrix,positioning,intersections}

\usepackage{babel}
\usepackage{xspace}
\newcommand{\ie}{\emph{i.e.}\xspace}
\newcommand{\eg}{\emph{e.g.}\xspace}
\newcommand{\cf}{\emph{cf.}\xspace}
\newcommand{\adhoc}{\emph{ad hoc}\xspace}
\newcommand{\Viceversa}{\emph{Vice versa}\xspace}
\newcommand{\viceversa}{\emph{vice versa}\xspace}

\theoremstyle{plain}\newtheorem{theorem}{Theorem}[section]
\theoremstyle{plain}\newtheorem{corollary}[theorem]{Corollary}
\theoremstyle{plain}\newtheorem{proposition}[theorem]{Proposition}
\theoremstyle{plain}
\theoremstyle{plain}\newtheorem{definition}[theorem]{Definition}
\theoremstyle{remark}\newtheorem{remark}[theorem]{Remark}
\theoremstyle{remark}\newtheorem{example}[theorem]{Example}

\usepackage{hyperref} % For hyperlinks in the PDF
\usepackage{cleveref} % cleaver handling of hyper references
\usepackage{cite}

\makeatletter
\newcommand{\superimpose}[2]{%
  {\ooalign{$#1\@firstoftwo#2$\cr\hfil$#1\@secondoftwo#2$\hfil\cr}}}
\makeatother

\makeatletter
	\newcommand{\footnoteref}[1]{%
		\protected@xdef\@thefnmark{\ref{#1}}\@footnotemark%
	}
\makeatother

\makeatletter
\newcommand{\customlabel}[2]{%
	\protected@write \@auxout {}{\string \newlabel {#1}{{#2}{\thepage}{#2}{#1}{}} }%
	\hypertarget{#1}{#2}
}
\makeatother

\newcommand{\defeq}{\triangleq}
\newcommand{\defiff}{\stackrel{\triangle}{\iff}}

\DeclareMathOperator{\cod}{cod}

\newcommand{\bottom}{\perp}

\makeatletter
\newcommand*{\boxwedge}{%
  \mathbin{%
    \mathpalette\@boxwedge{}%
  }%
}
\newcommand*{\@boxwedge}[2]{%
  \sbox0{$#1\boxplus\m@th$}%
  \dimen2=.5\dimexpr\wd0-\ht0-\dp0\relax % side bearing
  \dimen@=\dimexpr\ht0+\dp0\relax
  \def\lw{.06}% linw width as factor for height of \boxplus
  \kern\dimen2 % side bearing
  \tikz[
    line width=\lw\dimen@,
    line join=round,
    x=\dimen@,
    y=\dimen@,
  ]
  \draw
    (\lw/2,0) rectangle (1-\lw,1-\lw)
    (\lw,0) -- (.5,1-\lw-\lw/2) -- (1-\lw-\lw/2 ,0)
  ;%
  \kern\dimen2 % side bearing
}
\newcommand*{\boxvee}{%
  \mathbin{%
    \mathpalette\@boxvee{}%
  }%
}
\newcommand*{\@boxvee}[2]{%
  \sbox0{$#1\boxplus\m@th$}%
  \dimen2=.5\dimexpr\wd0-\ht0-\dp0\relax % side bearing
  \dimen@=\dimexpr\ht0+\dp0\relax
  \def\lw{.06}% linw width as factor for height of \boxplus
  \kern\dimen2 % side bearing
  \tikz[
	rotate=180,
    line width=\lw\dimen@,
    line join=round,
    x=\dimen@,
    y=\dimen@,
  ]
  \draw
    (\lw/2,0) rectangle (1-\lw,1-\lw)
    (\lw,0) -- (.5,1-\lw-\lw/2) -- (1-\lw-\lw/2 ,0)
  ;%
  \kern\dimen2 % side bearing
}
\makeatletter

\newcommand{\nmap}[2]{{#1\mapsto#2}}
\newcommand{\supp}[1]{\llfloor{#1}\rrfloor}
\newcommand{\triangleleftright}{\mathbin{\triangleleft\mspace{-1.5mu}\triangleright}}

\title{\vspace{-15mm}\fontsize{18pt}{10pt}\selectfont\textbf{
Well-Stratified Linked Data\\[5pt]for Well-Behaved Data Citation
}}
\author{
	{\renewcommand\footnotemark{}\thanks{Authors are organised by alphabetical order.}}%
	\textsc{\large Dario De Nart}\thanks{
		Lab.~of Artificial Intelligence,
		Department of Mathematics and Computer Science, University of Udine, Italy
	}\\
	\normalsize \href{mailto:dario.denarti@uniud.it}{\tt dario.denart@uniud.it}
	\and
	\textsc{\large Dante Degl'Innocenti}\footnotemark[1]\\
	\normalsize \href{mailto:deglinnocenti.dante@spes.uniud.it}{\tt deglinnocenti.dante@spes.uniud.it}
	\and
	\textsc{\large Marco Peressotti}\thanks{
		Lab.~of Models and Applications of Distributed Systems,
		Department of Mathematics and Computer Science, University of Udine, Italy
	}\\
	\normalsize \href{mailto:marco.peressotti@uniud.it}{\tt marco.peressotti@uniud.it}
}
\date{}

\hypersetup{
	colorlinks=true,
	linkcolor=blue!50!black,
	citecolor=blue!50!black,
	filecolor=blue!50!black,
	urlcolor=blue!50!black,
	pdfpagelayout=SinglePage,
	pdfpagemode=UseOutlines,
	pdftitle={Well-Stratified Linked Data for Well-Behaved Data Citation},
	pdfauthor={Dario De Nart, Dante Degl'Innocenti, Marco Peressotti},
	pdfsubject={Concurrency Theory},
	pdfkeywords={behavioural congruences, concurrent systems, quantitative models}
}

\begin{document}

\maketitle % Insert title

\thispagestyle{fancy} % All pages have headers and footers

%-------------------------------------------------------------------------------
%	ABSTRACT
%-------------------------------------------------------------------------------

\vspace{-5mm}
\begin{abstract}\noindent
In this paper we analyse the functional requirements of linked data citation and identify a minimal set of operations and primitives needed to realise such task. Citing linked data implies solving a series of data provenance issues and finding a way to identify data subsets. Those two tasks can be handled defining a simple type system inside data and verifying it with a type checker, which is significantly less complex than interpreting reified RDF statements and can be implemented in a non data invasive way. Finally we suggest that data citation should be handled outside of the data, possibly with an \adhoc language.
\end{abstract}

%-------------------------------------------------------------------------------
%	ARTICLE CONTENTS
%-------------------------------------------------------------------------------

\begin{multicols}{2}

\section{Introduction}

\looseness=-1
Over the last years data has become a more and more critical asset both in research and in application. While there is a general agreement on the need for data citation to ensure research reproducibility and to facilitate data reuse, the research community is still debating how to concretely realise it. 
Citing data is not a trivial task since it has a few notable differences from citing literature: data evolve over time, data availability might change over time, only a subset of data might be relevant, and on top of that the authorship of data is not always clear since it may be the result of an automated process (\eg sensor data), involve a large number of contributors (\eg crowdsourcing), or even be built on the top of other data (\eg inferring a taxonomy from a document corpus).
Levering on the insights provided by  \cite{altman2015introduction}, \cite{altman2013evolution}, \cite{silvello2015methodology}, and \cite{zhao2010provenance} we outline the following Data Citation functional requirements:
\begin{itemize}
\item \emph{Identification and Access}: Data Citation should provide a persistent, machine readable, and globally unique identifier for data; Moreover a reference to a persistent repository should also be provided to facilitate data access.
\item \emph{Credit and Attribution}: Data citation should facilitate giving credit and legal attribution to all contributors to the data. Such contributors might be humans as well as automated processes such as reasoners;
\item \emph{Evolution}: Data Citation should provide a reference to the exact version of the cited data, since data might change over time. This is a fundamental requirement for research reproducibility purposes.
\end{itemize}
An additional, non functional requirement, is efficiency: the data citation should lead to the data in practical time, which means fast enough for the purposes of data consumer applications. For instance a database query allows to access the data in practical time, while solving a complex set of logical clauses probably does not.

In the last years Linked Data has rapidly emerged as the preferred format for publishing and sharing structured data, creating a vast network of interlinked datasets \cite{heath2011linked}. 
However the open nature of the format makes data provenance hard to track, moreover the RDF Recommendation does not provide a clear mechanism for expressing meta-information about RDF documents. Semantic Web technologies such as OWL, RDF, and RDFS leverage upon description logic and first order logic and it is well known that an incautious usage of their primitives may lead to non decidable sets of conditions \cite{horrocks2003shiq}. 

With respect to the requirements of a good data citation expressed above, the Semantic Web community has proposed a number of solutions to the data provenance problem which addresses the problem of assessing the authorship of data. Methods for partitioning RDF graphs have been proposed as well and also version identification and storage of RDF data have already been discussed. However  most of those solutions imply the embedding of meta-information inside RDF data.
This practice tends to make data cumbersome and the usage of reification \cite{hayes2004rdf} to realise tasks such as generating data subsets may lead to a combinatorial explosion of triples.

In this paper we discuss a simple framework to satisfy data citation requirements levering on the stratification of linked data, which basically means providing a separation between proper data and meta-information. Such separation can be effectively guaranteed with the usage of a simple type system allowing programs such as reasoners to discriminate in an efficient way.
We would also like to show that the fact that Linked Data technologies such as RDF and OWL are powerful enough to let you seamlessly represent and embed meta-information inside the data does not mean that you really \emph{should}.

\paragraph{Synopsis}
\Cref{sec:related} complements the introduction with a concise survey of related work and pointers for the interested reader. \Cref{sec:formalisation} recalls an established formalisation of meta-information over linked data that abstracts from implementation dictated details and conventions. This formalisation is completed with a simple yet expressive algebra and an abstract notion of reasoner which offer the formal context for developing the paper main contribution. Said contribution is introduced in \Cref{sec:coherent} and is a notion of \emph{coherence for (meta)information} capable of characterising properties of information organisation. This result is augmented with the notion of \emph{coherent reasoner} \ie an abstract characterisation of reasoners that can operate on coherent data while preserving such cornerstone property. Finally, we investigate the cost of verifying and preserving coherency:
\begin{itemize}
	\item
	In \Cref{sec:ws-with-types} we introduce a type system that allows us to reduce coherence checking to type checking; type inference can be applied to derive type-annotations in order to speed-up future checks.
	\item
	In \Cref{sec:ws-with-graphs} we reduce coherence checking to suitable graph problems thus deriving an algorithm that can asses whether a given data store is coherent during a single read \ie in linear time. 
 	\item 
	This verification algorithm proceeds incrementally; we show how, given a coherent data store, the algorithm can check whether an operation preserves coherency. Remarkably, this on-the-fly check has negligible costs: linear to the number of entries created, deleted or updated.
\end{itemize}
Each subsection is completed with a short \emph{takeaway message} paragraph containing general remarks and intuitions aimed at rendering the technical results more accessible.
In \Cref{sec:language} we discuss how these results and notions may be translated into practice: in particular, we envision a new modelling layer with richer language support on top of existing technologies such as RDF.
Final remarks are provided in \Cref{sec:conclusions}.

\section{Related Work}
\label{sec:related}

Data citation has already been explored by the Semantic Web community and it significantly overlaps with the problem of assessing data provenance since determining the authorship of data is vital for citation purposes and both tasks need meta-information over data.
Provenance has already been widely discussed by the Semantic Web community leveraging on the insights provided by the Database community \cite{buneman2001and}.
Provenance information can be represented exploiting two approaches: the annotation approach and the inversion approach \cite{omitola2010provenance}. In the first approach all meta-information is explicitly stated, while in the latter is computed when needed in a lazy fashion which requires external resources containing the information upon which provenance is entailed to be constantly available. The annotation approach is favoured since it provides richer information and allows data to be self-contained; several vocabularies have been proposed to describe meta-information over linked data such as \emph{VoID} (Vocabulary of Interlinked Datasets) \cite{alexander2011describing}, that offers a rich language for describing Semantic Web resources built on top of well known and widely used ontologies such as foaf\footnote{\url{http://xmlns.com/foaf/spec/}} and Dublincore\footnote{\url{http://dublincore.org/documents/dcmi-terms/}}, and \emph{PROV Ontology} (PROV-O)\footnote{\url{http://www.w3.org/TR/prov-o/}}, which is the lightweight ontology for provenance data standardized by the W3C Provenance Working Group.
Regardless of the vocabulary used, adopting the annotation approach will result in producing a lot of meta-information which might exceed the actual data in size: provenance data in particular increases exponentially with respect to its depth \cite{simmhan2005survey}.
For more information about the problem of data provenance, we reference the curious reader to \cite{anam2015linked}.
The state of the art technique for embedding meta-information in RDF, is reification \cite{zhao2010provenance} which consists in assigning a URI to a RDF triple by expressing it as an \emph{rdf:Statement} object. 
Recently the RDF 1.1 Recommendation \cite{world2014rdf} introduced the so called ``RDF Quad Semantic'' which consists in adding a fourth element to RDF statements which should refer to the name of the graph which the triple belongs to. The actual semantic of the fourth element however is only hinted, leaving room for interpretation and therefore allowing semantics tailored to fit application needs.
In \cite{silvello2015methodology} is presented a methodology for citing linked data exploiting the quad semantics: the fourth element is used as identifier for RDF predicates allowing the definition of data subsets. Other usages of the fourth element include specification of a time frame, uncertainty marker, and provenance information container \cite{carroll2005named}.
Finally, the idea of using a type system to ease the fruition of semantic resources is not new to the Semantic Web community: the authors of\cite{ciobanu2014minimal} propose a type system to facilitate programmatic access to RDF resources.

\section{Formalisation}
\label{sec:formalisation}
In this section we provide a uniform formalisation of meta-information over linked data and of the related operations. Our formalisation abstracts over implementation details (such as how triples are extended with a fourth element) retaining all relevant information.

\subsection{Families of Named Graphs}
\looseness=-1
A \emph{Named Graph} (herein NG) is a labelled set of triples and may consist in a single labelled statement or in a larger set including multiple statements. Usually in RDF data triples are labelled using \emph{reification}, however to the extents of our discussion how the triples are labelled is irrelevant. Following the formalisation proposed in \cite{carroll2005named} we define a family of NGs as a 5-tuple $\langle N,V,U,B,L \rangle$ where:
\begin{compactenum}[(a)]
\item $U$ is a set of IRIs;
\item $L$ is a set of literals;
\item $B$ is a set of blank nodes;
\item $V$ is the union of the pairwise disjoint $U$, $L$, and $B$;
\item $N$ is a set of assignments $u \mapsto (v,u',v')$ 
mapping each $u \in U$ to at most one triple $(v,u',v') \in V \times U \times V$.
\end{compactenum}
Equivalently, the set $N$ can be read as a partial function 
$n\colon U \rightharpoonup V \times U \times V$ called 
\emph{naming function} of the NG family.
The set $V$ is called \emph{vocabulary} of the NG family for it defines all the
IRIs, literals etc.~appearing in the family.

Note how every pair in $N$ consists in a label (the first element) and a non void RDF graph (the second element), thus is an NG. 
It is important to stress how the above formalisation 
holds regardless of the actual technique employed to 
associate an identifier to a named graph.

Intuitively, assigning an IRI to a triple puts that IRI in the r\^ole of
\emph{meta-information} w.r.t.~that triple whence thought as
\emph{information}. Note that the separation between information and meta-information is not absolute but \emph{relative} to the context \ie the level where the reasoning happens. For an example, consider the RDF snipped:
\begin{Verbatim}[
	tabsize=3, 
	gobble=1,
	xleftmargin=3ex,
	commandchars=\\\{\},
	codes={\catcode`$=3\catcode`^=7\catcode`_=8}]
	$x$ type      statement
	$x$ subject   $y$
	$x$ predicate $b$
	$x$ object    $c$
	$y$ type      statement
	$y$ subject   $a$
	$y$ predicate $b$
	$y$ object    $c$	
\end{Verbatim}
Accordingly to the reification semantics, here $x$ is assigned
to triple $(y,b,c)$ and $y$ to the triple $(a,b,c)$ hence,
can be seen as an NG family whose $N$ is corresponds to the assignments:
\[
	\nmap{x}{(y,b,c)}\qquad\text{and}\qquad
	\nmap{y}{(a,b,c)}\text{.}
\]
Clearly, $y$ plays the r\^ole of meta-information with respect to
the triple $(a,b,c)$ and
$x$ plays the r\^ole of meta-information about $(y,b,c)$
whence $(a,b,c)$.
W
\begin{remark}[Blank nodes]
	\label{rem:blank-nodes}
	Since the presence of blank nodes arises several non trivial problems for the purposes of merging and comparing linked data, the last W3C recommendation suggests the replacements of blank nodes with IRIs\footnote{\url{http://www.w3.org/TR/rdf11-concepts/\#section-blank-nodes}} when data references are expected. Linked data including blank nodes are still common on the Web, but one can always assume the existence a renaming function $r\colon B \rightharpoonup U$ assigning unique IRIs to blank nodes.
\end{remark}

\looseness=-1
Above we recalled NG families as introduced in \cite{carroll2005named}
however, because of Remark~\ref{rem:blank-nodes},
we could equivalently characterise them by their associated
naming function alone. 
Clearly, every 5-tuple $\langle N,V,U,B,L \rangle$ 
describing a family of NGs uniquely induces a naming function $n$
but the opposite is not true in general. This can be traced down to
the later lacking information the separation of $V\setminus U$ into
blank nodes and literals. As a consequence of Remark~\ref{rem:blank-nodes}
we can safely assume $B = \emptyset$ and hence recover $L$ as $V\setminus U$.
Formally, for a naming function $n\colon U \rightharpoonup V \times U \times V$,
define the associated NG family as follows:
\begin{enumerate}[(a)]
	\item
		$N$ is the function graph of $n$ \ie 
		\[N \defeq \{(u,(v,u',v')) \mid n(u) = (v,u',v')\}\text{;}\]
	\item
		$V$ is the first or third projection of $n$ codomain \ie 
		\[V \defeq \pi_1(\cod(n)) = \pi_3(\cod(n))\text{;}\]
	\item
		$U$ is the second projection of $n$ codomain \ie 
		\[U \defeq \pi_2(\cod(n))\text{;}\]
	\item
		$L$ is the difference $V \setminus U$;
	\item
		$B$ is empty.
\end{enumerate}
We use the two presentations interchangeably.

\paragraph{Takeaway message}
Along the lines of  \cite{carroll2005named}, we described an abstract 
formalisation encompassing RDF data triples labelled
using reification, among other equivalent representations.
The main reason behind this effort is to have a disciplined
representation that:
\begin{compactenum}[(a)]
\item 
	avoids the use of reification while retaining its 
	expressive power;
\item
	hides conventions and idiosyncrasies due to 
	implementation details;
\item
	clearly separates information from meta-information.
\end{compactenum}

\subsection{A simple algebra for NG families}
\label{sec:ng-algebra}
Families of named graphs can be organised into the partial 
order\footnote{%
	In general, even if we fix the vocabulary, NG families 
	may form a proper class. Because of the scope of this work 
	we shall avoid the technicality of partially ordered classes.
}
defined as:
\[
	n \sqsubseteq n' \defiff n(u) = (a,b,c) \implies n'(u) = (a,b,c)\text{.}
\]
Intuitively, $n \sqsubseteq n'$ means that $n'$ has more data than $n$ but
does not carry any implication on the semantic of such information since,
for instance, the former may contain semantically inconsistent information
not held by the latter.
We say that $n'$ \emph{is an extension of} $n$ whenever $n \sqsubseteq n'$.
Note that $n$ and $n'$ are not required to have the very same
set of IRIs and literals (hence vocabulary).
In fact, for any  pair of NG families $n\colon U \rightharpoonup V \times U \times V$ and
$n'\colon U' \rightharpoonup V' \times U' \times V'$,
the only information about $U$, $U'$, $V$, and $V'$
we can infer from $n \sqsubseteq n'$ is that
$U'$ and $V'$ overlap with $U$ and $V$ where $n$ is defined. Formally:
\[
	n(u) = (a,b,c) \implies u, b \in U' \land a,c \in V'\text{.}
\]
Two families are equivalent iff they extend each other:
\[
	n \equiv n' \defiff n \sqsubseteq n' \land n \sqsupseteq n'\text{.}
\]
\looseness=-1
Clearly, $\equiv$ is an equivalence relation and all of its equivalence classes
have a canonical representative family \ie the unique and minimal element
in the class. Intuitively, to obtain such family we only need to start with
any family in the equivalence class and trim its vocabulary and set of IRIs
by removing any element that does not appear in the naming function range or image
\ie remove everything but $u$, $a$, $b$, $c$ such that $n(u) = (a,b,c)$.
Empty families form an equivalence class $\varnothing$ represented by 
the unique naming function for the empty vocabulary.
Hereafter, we shall not distinguish families in the same equivalence class, unless
otherwise stated.

A NG family is called \emph{atomic} whenever it assigns
exactly one name \ie a naming function $n$ that
defined exactly on one element of its domain.
We shall denote atomic families by their only assignment:
\[
		(\nmap{x}{(a,b,c)})(u) \defeq 
		\begin{cases}
			(a,b,c) & \text{if } x = u \\
			\bottom & \text{otherwise}
		\end{cases}
\]
which corresponds to the RDF triple $a, b, c$ plus the fourth element $x$ that identifies the named graph. 

From the order-theoretic perspective,
atomic NG families are atoms for the order $\sqsubseteq$.
In fact, $n$ is atomic iff:
\[
	m \sqsubseteq n \implies m \equiv n \lor m \equiv \varnothing
	\text{.}
\]
For the sake of simplicity, we shall abbreviate 
$\nmap{x}{(a,b,c)}$ as $(a,b,c)$ when the particular choice of 
$x$ is irrelevant and confusion seems unlikely; we still 
assume $x$ to be implicitly unique in the context of its use.

In general, any non-empty set $S$ of NG families admits a minimal
element $\sqcap S$ given by the intersection
of all the families in $S$ or, using the naming function
presentation, to:
\[
	(\sqcap S)(u) = \begin{cases}
		(a,b,c) & \text{if } \forall n \in S\ n(u) = (a,b,c)\\
		\bottom & \text{otherwise}
	\end{cases}
\]
from which it is immediate to see that ${\sqcap S} \sqsubseteq n$
for any $n \in S$.
We may denote $\sqcap\{n_1,\dots,n_k\}$ by $n_1 \sqcap n_2 \sqcap \dots \sqcap n_k$ when confusion seems unlikely.

It could be tempting to dualise $\sqcap$ defining maximal elements as
follows:
\[
	(\sqcup S)(u) = \begin{cases}
		(a,b,c) & \text{if } \exists n \in S\ n(u) = (a,b,c)\\
		\bottom & \text{otherwise}
	\end{cases}
\]
however, note that this represents the merging of heterogeneous data
and may not yield an NG-family, hence is
not a total operation. For instance, consider any set containing
some $n$ and $n'$ assigning some $u$ to different triples
\eg: $n(u) = (a,b,c)$ and $n(u) = (a',b,c)$ where
$a\neq a'$.
The operation can be turned into a total one 
without altering the definition of NG family by adopting some resolution policy for conflicting overlaps:
discarding both, either, introducing suitable renaming, or more
sophisticated techniques developed under the name of 
\emph{ontology alignment}\footnote{%
    Ontology alignment is a vast and debated topic that has been
    formalised in several ways, we address the curious reader to
    \cite{zimmermann2006formalizing}.
}
a complex task whose formalisation exceeds the scope of this work.
In the remaining of this subsection we show how the aforementioned
simple conflict resolution policies yield ``surrogate'' operators for 
$\sqcup$ corresponding to specific and well-known operations.

Before defining these operations let us introduce some
auxiliary definitions and notation.
For an NG family $n\colon U \rightharpoonup V \times U \times V$,
its \emph{support} is the subset $\supp{n}$ of $U$ whose elements
are assigned to some graph by $n$:
\[
    \supp{n} \defeq \{u \mid n(u) \text{ is defined}\}
    \text{.}
\]
For a pair of families $n_1$ and $n_2$, their 
\emph{conflict set} $n_1 \lightning n_2$ 
is the set of IRIs being assigned to different graphs 
by $n_1$ and $n_2$, \ie:
\[
	n_1 \lightning n_2 \defeq 
	\{u \mid  
		\supp{n_1} \cap \supp{n_2} \land 
		n_1(u) \neq n_2(u)
	\}
	\text{.}
\]
For a family $n$, an injective map $\sigma$ from its vocabulary
defines a renamed family $n[\sigma]$:
\[
	n[\sigma](u) = 
	\begin{cases}
		(\sigma(a),\sigma(b),\sigma(c))  & \text{if } n(u') = (a,b,c) \land \sigma(u') = u\\
		\bottom & \text{otherwise} 
	\end{cases}
\]
We adopt the convention that any renaming $\sigma\colon V \to V'$
implicitly extends to a renaming on any superset of $V$ that
acts as $\sigma$ on $V$ and as the identity elsewhere.

The first ``surrogate'' for $n_1 \sqcup n_2$
resolves conflicts by ignoring
any conflicting information from $n_2$. Formally:
\[ 
	(n_1 \triangleright n_2)(u) \defeq
	\begin{cases}
		n_1(u) & \text{if } u \in \supp{n_1}\\
		n_2(u) & \text{if } u \in \supp{n_2}\setminus\supp{n_1}\\
		\bottom & \text{otherwise}
	\end{cases}
\]
We call this operation a surrogate for binary joins
since it is total and agrees with joins whenever
they are defined:
\[
	n_1 \triangleright n_2 = n_1 \sqcup n_2\text{.}
\]

The operation $\triangleright$ is associative:
\[
    n_1 \triangleright (n_2 \triangleright n_3)
    = (n_1 \triangleright n_2) \triangleright n_3
    \text{,}
\]
has the empty family as a left and right unit:
\[
    n \triangleright \varnothing = n = \varnothing \triangleright n
	\text{,}
\]
is idempotent:
\[
    n \triangleright n = n
    \text{,}
\]
and hence defines an idempotent monoid of NG families.
In general, $\triangleright$ is not commutative:
it is easy to check that commutativity does not hold 
unless $n_1$ and $n_2$ agree wherever both are 
defined \ie the set of conflicts 
$n_1 \lightning n_2$ is empty.
The operation $\triangleright$ is monotonic:
\begin{gather*}
    n_1 \sqsubseteq n_2
    \implies
	n_1 \triangleright n_3
	\sqsubseteq
	n_2 \triangleright n_3
	\\
    n_1 \sqsubseteq n_2
    \implies
	n_3 \triangleright n_1
	\sqsubseteq
	n_3 \triangleright n_2
\end{gather*}
and distributes over meets:
\[
    n_1 \sqcap (n_2 \triangleright n_3)
    = 
    (n_1 \sqcap n_2) \triangleright (n_1 \sqcap n_3)
    \text{.}
\]

The second ``surrogate'' for $n_1 \sqcup n_2$
uses an approach symmetric to $\triangleright$
by discarding conflicting data from its first operand.
Because of the symmetry, it will be denoted as $\triangleleft$
and presented simply mirroring $\triangleright$:
\[ 
	n_1 \triangleleft n_2 \defeq n_2 \triangleright n_1
	\text{.}
\]
Like $\triangleright$, this operation yields an idempotent
monoid over NG families; moreover, both operations share 
the same unit and distribute over each other.

A third option is to resolve conflicts by discarding every data 
assigned to IRIs in the conflict
\[ 
	(n_1 \triangleleftright n_2)(u) =
		\begin{cases}
		n_1(u) & \text{if } u \in \supp{n_1}\setminus (n_1 \lightning n_2)\\
		n_2(u) & \text{if } u \in \supp{n_2}\setminus (n_1 \lightning n_2)\\
		\bottom & \text{otherwise}
	\end{cases}
\]
Although this operation may seem more drastic than $\triangleright$
and $\triangleleft$, $\triangleleftright$ treats its operands equally and
coincides with $\sqcup$ whenever the later is defined.
As a consequence, $\triangleleftright$ is commutative, associative,
idempotent and has a unit.

Note that the conflict set between $(n_1 \triangleleftright n_2)$
and $n_1$ or $n_2$ is always empty and hence the following
equations are well-defined and hold true:
\begin{gather*}
	n_1 \triangleright n_2 = (n_1 \triangleleftright n_2) \sqcup n_1
	\ 
	\\
	n_1 \triangleleft n_2 = (n_1 \triangleleftright n_2) \sqcup n_2
	\text{.}
\end{gather*}
\Viceversa, $\triangleleftright$ can be derived from
$\triangleright$ and $\triangleleft$:
\[ 
	n_1 \triangleleftright n_2 = (n_1 \triangleleft n_2) \sqcap (n_1 \triangleright n_2)
	\text{.}
\]

We introduce a surrogate for $n_1 \sqcup n_2$ that handles conflicts without
discarding any information from its operands by
renaming all conflicting assignments made by 
$n_1$ and $n_2$\footnote{
For those familiar with OWL2, this strategy bears significant similarity with \emph{punning} which is an implicit renaming of conflicting entities, \eg a class and an individual or an object property and a data-type property sharing the same IRI.
}. 
Intuitively, this may be though to
be implemented by ``doubling'' the conflict set as:
\[
	(n_1\lightning n_2)\times\{1,2\} =
	\{\langle u, i \rangle\mid u \in n_1\lightning n_2 \land i \in\{1,2\}  \}
\]
In general, we abstract from specific renaming policies
by assuming a pair of injective maps
$\sigma_1$ and $\sigma_2$ such that the
conflict set $(n_1[\sigma_1]\lightning n_2[\sigma_2])$ is empty.
Referring to the intuitive implementation described above,
each renaming $\sigma_i$, for $i$ in $\{1,2\}$, is defined as follows:
\[
	\sigma_i(u) = \begin{cases}
		\langle u, i\rangle & \text{if } u \in (n_1\lightning n_2)\\
		u & \text{otherwise}
	\end{cases}
\]
Then, the last operator is defined as
\[
	(n_1 \boxvee n_2) \defeq n_1[\sigma_1] \sqcup n_2[\sigma_2]
	\text{.}
\]
The binary join $n_1[\sigma_1]\sqcup n_2[\sigma_2]$ is always well-defined
since $(n_1[\sigma_1]\lightning n_2[\sigma_2])$ is empty by
assumption on the renaming maps $\sigma_1$ and $\sigma_2$
and $n_1[\sigma_1]\sqcup n_2[\sigma_2] = n_1 \sqcup n_2$
whenever $n_1 \lightning n_2 = \emptyset$.

We extend the conventions introduced for $\sqcap$ to the
operations $\triangleleft$, $\triangleright$, $\triangleleftright$, and $\boxvee$:
we shall use $\boxvee\{n_1,\dots,n_k\}$ for
$n_1 \boxvee n_2 \boxvee \dots \boxvee n_k$ and \viceversa.

Every NG family can be defined using the constants and operations described in this subsection.
More complex operations such as inferences are described in the remaining part of this section.

\paragraph{Takeaway message}
In this subsection we introduced a simple yet expressive algebra
for describing NG families. Although, more convenient constructs,
operations or sophisticated techniques are not part of this algebra,
we believe they can be easily implemented on top of it hence suggesting
the use of this algebra as a \emph{core} language for targeting
(via compilation) any chosen implementation of NG families such as
reified RDF \cite{hayes2004rdf} or graphical models like \cite{montanari:2010gs-graphs,milner:bigraphbook,Engels1995101}.

\subsection{Reasoners over NG families}
In this subsection we formalise the problems of provenance, subsetting, and
versioning in the setting of NG families as formalised above.
To this end, we introduce an abstract and general notion of \emph{abstract reasoner} subsuming any process (automatic or not) transforming NG families. 
At this level of abstraction we formalise the problem of
understanding whether 
\vspace{-2ex}\begin{quote}
	``$x$ in $n$ has been generated by $\gamma$''
\end{quote}\vspace{-2ex}
for an IRI $x$, a named graph family $n$ and reasoner $\gamma$
and then show how provenance, subsetting, and versioning 
are covered as instances of the above.

\begin{definition}
	\label{def:abstract-reasoner}
	An \emph{(abstract) reasoner} is a
	(partial) function over NG families.
\end{definition}

For a simple example, consider a reasoner
that expands a family $n$ by computing the 
transitive and reflexive closure of any
predicate $b$ that $n$ describes as being transitive 
or reflexive \eg by means of some graph 
$(b,\mathtt{predicate},\mathtt{transitive})$.
Such reasoner can be described
as assigning to any family $n$
the least family closed under the 
derivation rules:
\begin{gather*}
	\frac{n(x) = (a,b,c)}{n \vdash (a,b,c)}
	\\
	\frac{
		n \vdash (a,b,c)
		\quad
		n \vdash (c,b,d)
		\quad
		n \vdash (b,\mathtt{predicate},\mathtt{transitive})
	}{
		n \vdash (a,b,d)
	}
	\\
	\frac{
		n \vdash (b,\mathtt{predicate},\mathtt{reflexive})
	}{
		n \vdash (a,b,a)
	}
\end{gather*}
Hence $\gamma(n) = \{(a,b,c) \mid n \vdash (a,b,c)\}$.

Likewise, symmetry can be computed by means of:
\[
	\frac{
		n \vdash (b,\mathtt{predicate},\mathtt{symmetric})
		\quad
		n \vdash (a,b,c)
	}{
		n \vdash (c,b,a)
	}	
\]
and reversible predicates by means of:
\begin{gather*}
	\frac{
		n \vdash (b,\mathtt{reverse},\overline{b})
		\quad
		n \vdash (a,b,c)
	}{
		n \vdash (c,\overline{b},a)
	}	
	\\
	\frac{}{n \vdash (\mathtt{reverse},\mathtt{predicate},\mathtt{symmetric})}
\end{gather*}
These examples describe \emph{monotonic} reasoners since 
$n \sqsubseteq \gamma(n)$ for any family $n$;
however, an abstract reasoner may be non-monotonic as well: consider for instance a
human annotator performing a revision of an ontology, such an annotator is 
 likely to both add and withdraw triples from the knowledge base and still fits the definition
 of an abstract reasoner.
A simple example of situation where the withdrawal of a triple is needed to
keep data consistent is offered by the derivation rule:
\begin{gather*}
	\frac{
		n \not\vdash (a,\mathtt{after},b)
	}{
		n \vdash (a,\mathtt{first},\mathtt{event})
	}
\end{gather*}

All examples of reasoners described so far essentially work at the
triple level for they never really follow any IRI.
This kind of reasoner never crosses the boundary between 
information and meta-information but this is not true in
general. Actually, crossing such boundary is often necessary
when reasoning about the (meta)information stored in an NG family
as we introduce reasoners that handle operations over meta-information such as
tracking data authorship, extracting data subsets, and managing versions.
Before we delve into this topic, let us introduce
some auxiliary reasoners and definitions
that allow us to describe such reasoners as well.

For an abstract reasoner $\gamma$ and a family $n$, define the sets of created, updated, and deleted assignments as:
\begin{align*}
 C_\gamma(n) & \defeq \supp{\gamma(n)} \setminus \supp{n}\\ 
 U_\gamma(n) & \defeq \{x \mid x \in \supp{\gamma(n)}\cap\supp{n} \land 
 	n(x) \neq \gamma(n)(x) \}\\
 D_\gamma(n) & \defeq \supp{n}\setminus\supp{\gamma(n)}
\end{align*}
Given $n$ and $\gamma$, computing the above sets could be
prohibitively demanding even under the assumption of
NG families being finite.
In practice, changes are recorded by explicitly tagging
all affected assignments. Since this good practice is
not imposed by Definition~\ref{def:abstract-reasoner}
we introduce new reasoners that extend the output of
any given reasoner $\gamma$ with this tagging information.
\begin{align*}
 \Delta^C_\gamma(n) & \defeq \gamma(n) \boxvee \{(\gamma,\mathtt{new},x) \mid x \in C_\gamma(n) \}\\
 \Delta^U_\gamma(n) & \defeq \gamma(n) \boxvee \{(\gamma,\mathtt{upd},x) \mid x \in U_\gamma(n) \}\\ 
 \Delta^D_\gamma(n) & \defeq \gamma(n) \boxvee \{(\gamma,\mathtt{del},x) \mid x \in D_\gamma(n) \}
\end{align*}

Everything added by one of the above reasoners
to the output of $\gamma$ is meta-information with respect to said output.
As expected, this data enables reasoning on the evolution of the family itself.
For instance, consider the following simple inference that
reconstructs which reasoner used other reasoners
by looking whether they recorder their activity:
\begin{gather*}
	\frac{
		n \vdash (\gamma,\mathtt{new},y)
		\quad
		n(y) = (\delta,q,z)
		\quad
		p,q \in \{\mathtt{new},\mathtt{upd}\}
	}{
		n \vdash (\gamma,\mathtt{uses},\delta)
	}
\end{gather*}
\looseness=-1
Clearly, any reasoning on the evolution of a named family due to the action of reasoners (via the oversight of $\Delta$, for exposition convenience) inherently require that some references stored as named graphs are followed hence that the boundary between information and meta-information is crossed at some point. Levering on these definitions, one can easily describe reasoners that realise authorship attribution over data with different granularity, reasoners to label and extract data subsets and versions.

\paragraph{Takeaway message}
In this section, we unearthed the core issues arising from
reasoning about information and meta-information as well and provided
a framework for describing all operations needed to perform data citation.
Since our formal treatment has been carried at the abstract level of
named graph families, we have shown how these issues
are inherent to the problem and independent from implementation
details and specific techniques such as reification.

\section{Coherent (meta)information}
\label{sec:coherent}
In this section we characterise a class of NG families
called \emph{well-stratified}
with the fundamental property of stratifying meta-information
over information in a way that prevents any infinite chain of
``downward'' references where the direction is interpreted as
crossing the boundary between meta-information and information.
Since practical NG families (hence triple stores) contain only
a finite amount of explicit information, absence of such chains
corresponds to the absence of reference cycles like, for instance,
in the NG family:
\[
	\nmap{x}{(y, b, c)}
	\sqcup
	\nmap{y}{(x, b, c)}
	\text{.}
\]
We characterise a class of abstract reasoners,
called \emph{coherent}, that preserve well-stratification
of named graph families they operate on.
Finally, we introduce a decidable and efficient procedures for assessing the well-stratification
of an NG family and operations on them.

\subsection{Well-founded relations}
Before we define well-stratified NG families, let us recall some
auxiliary notions and notation.
A binary relation $R$ on a (non necessary finite) set $X$
is called \emph{well-founded} whenever every non-empty
subset $S$ of $X$ has a minimal element \ie there exists
$m \in S$ that is not related by $s \mathbin{R} m$ for $s \in S$.
This means that we can intuitively walk along $R$ going
from right to left for finitely many steps \ie we have to stop, eventually.
In fact, well-foundedness can be reformulated say that
$R$ contains no (countable) infinite descending chain
(\ie an infinite sequence $x_0,x_1,x_2,\dots$ such that 
$x_{n+1} \mathbin{R} x_{n}$).

\vspace{1ex}
\begin{example}
	The predecessor relation $\{(x,x+1) \mid x \in \mathbb{N}\}$ on the set 
	of natural numbers is well-founded.
	The prefix and suffix relations on the set $\Sigma^*$
	of finite words over the alphabet $\Sigma$
	is well-founded.
	Any acyclic relation on a finite set is (trivially) 
	well-founded.
	Point-wise and lexicographic extensions of well-founded relations
	are well-founded.
\end{example}

\looseness=-1
This kind of relations are common-place in mathematics and computer science since they provide the structure for several inductive and recursive principles and, with regard to this work aim, approaches for proving termination. Intuitively, the idea is to equip the state space of an algorithm with a well-founded relation and then show that each step of the algorithm travels such relation right to left (descends). By well-foundedness hypothesis, all descent paths are bound to terminate in a finite number of steps.

\vspace{-2ex}
\paragraph{Takeaway message}
Well-founded relations such as the successor relation over natural numbers are at the core of several techniques used to prove termination. Such techniques revolve around the idea of reducing the problem under scrutiny to walks along said well-founded relation and hence termination follows by the fact that any such (descending) walk cannot be infinite.

\subsection{Well-stratified NG families}
The intuitive desiderata of an NG family being free
from infinite paths descending along the (meta-)information chain is formally captured by the following definition.

\begin{definition}
    A family of named graphs is called \emph{well-stratified}
    whenever it comes equipped with a well-founded
    relation ${\prec}$ on its vocabulary s.t.~$n$
    descends along $\prec$ \ie 
	\[
		n(u) = (a,b,c) \implies 
		u \succ a \land  
		u \succ b\land
		u \succ c
		\text{.}
	\]
	The relation $\prec$ is called \emph{witness}
	for $n$.
\end{definition}
\vspace{-1ex}

Following any chain of assignments $\nmap{x}{a,b,c}$ described by a well-stratified NG family has to eventually terminate since each step corresponds to a step along $\prec$ which is well-founded by hypotheses.
Thus, any reasoner based on such visits is bound to terminate as long as it descends along $\prec$ and each internal step in its chain is decidable. Moreover, for a given NG family the length of these chains is known and bounded.

Operations described in Section~\ref{sec:ng-algebra}
preserve well-stratification under the assumption that 
all operands can share their witness $\prec$.

Abstract reasoners may easily break well-stratified stores. Intuitively most reasoning tasks and well-engineered human
annotation processes should be coherent, however
breaking the well-stratification of data is subtle and can
be achieved even with monotonic reasoning.
For instance, consider a set of triples where there 
exists a triple $(y, \mathtt{type}, \mathtt{statement})$
labelled with some IRI $x$ and an abstract reasoner 
$\gamma$ that adds a new triple $(x, \mathtt{type}, \mathtt{statement})$
labelled as $y$. 
This insertion is totally legit if we are using reification
but introduces a circularity in the chain of meta data since
the family now contains the following assignments:
\[
	\nmap{x}{(y, \mathtt{type}, \mathtt{statement})}
	\qquad
	\nmap{y}{(x, \mathtt{type}, \mathtt{statement})}
\]
and hence is no more well-stratified.
 
\begin{definition}
	An abstract reasoner is called \emph{well-behaved}
	whenever it preserves well-stratification.	
\end{definition}

Reasoners for provenance, subsetting, and versioning
are well-behaved for they cross the boundary between information
and meta-information only in one direction: descent.

In general, assessing whether a reasoner is well-behaved may not
be immediate especially since Definition~\ref{def:abstract-reasoner}
describes them as ``black boxes transforming NG families''.
There are several ways for describing classes of 
abstract reasoners with different degrees of expressiveness.
Covering all of them is out of the scope of this work and
indeed impossible\footnote{%
	We leave to the reader the exercise of formalising an human reasoner and prove it well-behaved.
}, still derivation rules are a presentation fit for many
reasoners (like the ones described so far) and well known across the computer science community.
This approach allows to quickly inspect the ``internals of the box''
and statically prove a reasoner well-behaved.
Moreover, it is possible to define reasoners that are
``well-behaved by design''  by imposing suitable
syntactic constraints on these derivation rules akin
to rule formats developed in the field of concurrency theory
(\cf \cite{mp:2014ultras-gsos,mp:2016ultras-gsos,ks2013:w-s-gsos}).
Albeit interesting, this topic cannot be fully developed
in the context of this work.

\paragraph{Takeaway message}
In this subsection we characterised a class of NG families that stratify meta-information over information without creating incoherences such as loops. In general, reasoners may easily break this cornerstone property and treating them as black boxes prevents any practical attempt to statically check whether they really break well-stratification. However, with access to enough information about the internal working of a reasoner (\eg its description in terms of derivation rules), established formal methods can be applied to prove it well-behaved; even develop languages for creating reasoners guaranteed to be well-behaved. Remarkably, reasoners for provenance, subsetting, and versioning admit well-behaved implementations.

\subsection{Assessing well-stratification using types}
\label{sec:ws-with-types}
NG families share some similarities formal graphical languages
like bigraphs and hierarchical graphs. This observation suggest
to introduce a simple type system, along the line of 
\cite{Engels1995101}, with a special type whose 
inhabitants are exactly well-stratified NG families.
Then, to verify if a given family is well-stratified
it would suffice to check if it is well-typed.

For the aims of this work, we introduce a simple type
system whose only type $\checkmark$ is inhabited by
exactly well-stratified families. Judgements are of the form
\[\Gamma \vdash n \colon \checkmark\]
where $n\colon U \rightharpoonup V \times U \times V$ is
a family of named graphs and the stage $\Gamma$ is
a partial function from the vocabulary $V$ to a well-founded
structure. For instance, could map $V$ to the set of
natural numbers under the successor relation: 
$\Gamma\colon V \rightharpoonup \mathbb{N}$.

The proposed type system is composed by three typing rules:
\begin{gather*}
	\frac{
	}{
		\Gamma \vdash \varnothing\colon \checkmark
	}
	\\[5pt]
	\frac{
		\Gamma(x) > \Gamma(a)
		\quad
		\Gamma(x) > \Gamma(b)
		\quad
		\Gamma(x) > \Gamma(c)
	}{
		\Gamma \vdash \nmap{x}{(a,b,c)}\colon \checkmark
	}
	\\[5pt]
	\frac{
		\Gamma_1 \vdash n_1 \colon \checkmark
		\quad
		\Gamma_2 \vdash n_2 \colon \checkmark
		\quad
		\Gamma = \Gamma_1 \sqcup \Gamma_2
		\quad
		n = n_1 \sqcup n_2
	}{
		\Gamma \vdash n\colon \checkmark
	}
\end{gather*}
The first captures the fact that the empty family is always well-stratified.
The second ensures that $\Gamma$ describes relations on $V$ such that
the assignment $\nmap{x}{(a,b,c)}$ is well-stratified.
Finally, the third rule allows to break $n$ and $\Gamma$ reducing the
problem to smaller objects which can then be checked separately
(clearly, applying this rule with either $n_1$ or $n_2$ being $\varnothing$
is pointless).

Regardless of the structure used as codomain,
$\Gamma$ determines a class of relations
on $V$ that are well-founded
where $\Gamma$ is defined: for a relation 
$\prec$ on $V$ s.t.:
\[
	x \prec y \implies \Gamma(x) < \Gamma(y)
\]
the restriction of $\prec$ to $\supp{n}$
is clearly well-founded. This property is enough
to guarantee that any family $n$ such that
$\Gamma \vdash n \colon \checkmark$ (\ie is well-typed)
is well-founded. In fact, because of the above typing
rules, $\Gamma$ must be defined on every IRI and literal 
occurring in $n$ and hence $\prec$ as above is a witness
for $n$ being well-stratified.

We do not need to ``guess'' $\Gamma$. This function can be
obtained by applying the above typing judgements while considering
$\Gamma$ as an unknown collecting all the hypotheses on it
(\eg $\Gamma(x) > \Gamma(a)$ from the second rule)
in a set of constraints. 
Any partial function satisfying these constraints can be used
as $\Gamma$ to derive $\Gamma \vdash n \colon \checkmark$.
Although computing such solutions can be done pretty efficiently,
at this point it suffices to prove solution existence to
prove $n$ well-stratified.

In practice, type checkers may be helped by providing
typing annotations as separate meta-data, as primitives
of a specialised language for NG families, or just as
comments like in the following RDF snipped:
\begin{Verbatim}[
	tabsize=3, 
	gobble=1,
	xleftmargin=3ex,
	commandchars=\\\{\},
	codes={\catcode`$=3\catcode`^=7\catcode`_=8}]
	// $x$: 4; $y$: 2; $b$, $c$: 0
	$x$ type      statement
	$x$ subject   $y$
	$x$ predicate $b$
	$x$ object    $c$
\end{Verbatim}
This statically computed information can be used to
optimise reasoners since $\Gamma(x)$ provides an upped bound
to the length of meta-information/information steps starting from $x$.
As we show in Subsection~\ref{sec:ws-with-graphs}, the very same
information can be used to efficiently reject any operation that 
breaks well-stratification.

\paragraph{Takeaway message}
In this subsection we characterised well-stratified
NG families by means of a simple type system and showed
the type inference problem to be decidable. This approach
suggests to explore the use of more expressive type systems
and sortings in order to express/enforce richer properties
about NG families. Moreover, the connection between NG families
and formal graphical models suggests the possibility to
extend compositionality results such us those offered
by \emph{monoidal sortings} \cite{mp:br-tr13} to this settings. 

\subsection{Assessing well-stratification using graphs}
\label{sec:ws-with-graphs}

For a family $n$ whose support $\supp{n}$ is finite,
the only way to not be well-stratified is to
contain cycles of dependencies between information and
meta-information: the only way for a relation on a finite
set to not be well-founded is to contain cycles. 
Therefore, if we read assignments described by $n$
as arcs in a directed graph we can reduce the problem of 
checking if $n$ is well-stratified to checking if this 
``graph of dependencies'' is free from cycles.

\begin{definition}
	For an NG family $n$ its \emph{dependency graph} $G_n$
	is a graph with	$\supp{n}$
	and $\{(x,y_i) \mid 
				n(x) = (y_1,y_2,y_3)
				\land
				y_i \in \supp{n}
			\}$
	as its set of nodes and edges, respectively.
\end{definition}

For a family $n$ containing a finite amount of
\emph{explicit} meta-information (as in any
real-world scenario) well-stratification
can be checked with time cost linear to the
number of assignments described by $n$
(\ie the cardinality of the set $\supp{n}$).
\begin{proposition}
	For a family of named graphs $n$
	such that $\supp{n}$ is	finite,
	$n$ is well-stratified if and only if
	its dependency graph $G_n$ is a 
	directed acyclic graph.
\end{proposition}
\begin{proof}
	By hypothesis on $n$, $G_n$ has finitely many edges and nodes
	hence the only way for it to contain an infinite path
	is to have a directed cycle.
\end{proof}

Absence of directed cycles reduces to the existence of
a topological sorting which can be easily computed
in polynomial time with Tarjan's algorithm.
Intuitively, this amounts to a depth first visit of 
the dependency graph $G_n$: a graph whose nodes have at most three outgoing
edges, hence the time complexity actually is linear.

\begin{corollary}
	Well-stratification can be checked with a time cost
	linear in the size of $\supp{n}$.
\end{corollary}

Because of the size that real-world triple stores can reach,
computing a topological sorting of $G_n$ from scratch every
time an operation on $n$ is performed can be a daunting task.
In the remaining of this section we describe how to efficiently
and precisely reject all changes that will break well-stratification.

In Section~\ref{sec:ng-algebra} we described a core
algebra for NG families highlighting that complex 
transformations basically reduce to insertions and deletions
of name assignments (updates are modelled as atomic pairs of 
deletions and insertions, as usual). Clearly, deleting a named
assignment preserves well stratification and hence only insertions
need to be checked before being carried out.

A way to curb this cost is to cache the information about
on the topological sorting in a map from $\supp{n}$ to some
linear order relation on a dense but limited set such us
the rational part of the interval $[0,1)$. This order relation
is not well-founded but it is acyclic and hence any restriction
to a finite subset of $[0,1)$ is well-founded. Moreover, being
dense, we can always ``make room'' for newly inserted (meta)information.

In the following let $n$ be a well-stratified NG family and
let $m$ be a partial map from $U$ to the subset of rational numbers:
\[
	\{l \cdot 2^{-k} \mid k \in \mathbb{N} \land 0 \leq l < k \}
	\text{.}
\]
Under the assumption that we start from an empty family $n$
and an empty map $m$, the algorithm ensures that after
an arbitrary sequence of insertions
\begin{itemize}
\item
	$m$ is defined exactly on every element occurring in $n$
	as information or meta-information;
\item 
	the natural order on $\mathbb{Q}$ defines a well-founded relation on all 
	elements where $m$ is defined;
\item
	an operation is rejected if, and only if, it does not preserve well-stratification
	of $n$.
\end{itemize}
Note that the first two points imply well-stratification.

Consider the insertion of $\nmap{x}{(a,b,c)}$ in $n$.
Since $x \not\in \supp{n}$ we have to consider two main
scenarios: in the first $x$ does not occur in $n$ 
(\ie $m(x) = \bottom$) whereas in the second
$x$ occur in $n$ as information only ($m(x) \neq \bottom$ but $n(x) = \bottom$).

Assume $m(x) = \bottom$, we need to assign to $x$ a value above
those assigned to $a$, $b$, and $c$ but below everything that we already put
above these three piece of data:
\begin{align}
	\label{eq:new-x}
	m(x) & \defeq y + \frac{|y-z|}{2}\\
	\notag
	y & = \max\{m(a),m(b),m(c)\}\\
	\label{eq:new-z}
	z & = \inf\{w \mid y < w \land m(u) = w\}
\end{align}
Note that although the definition
of $z$ may seem a bit convoluted, it can be readily
implemented by means of an ordered set of values
from $m$: \eqref{eq:new-z} is exactly the first successor to $y$ in such structure and \eqref{eq:new-x} corresponds to an insertion right between $y$ and $z$.

Assume $m(x) \neq \bottom$, we have three sub-cases:
\begin{enumerate}[(a)]
\item 
	If $m(x) < \max\{m(a),m(b),m(c)\}$ then
	at least one of $a$, $b$, or $c$ 
	occurs in $n$ in the r\^ole of meta-information 
	w.r.t.~$x$ and thus $\nmap{x}{(a,b,c)}$ 
	is rejected and the algorithm terminates.
\item
	If $m(x) = \max\{m(a),m(b),m(c)\}$ 
	we need to promote $x$ ``pushing'' everything
	above it up and everything else down as in
	the first scenario; therefore,
	$m(x)$ is redefined using \eqref{eq:new-x}.
\item
	If $m(x) > \max\{m(a),m(b),m(c)\}$ then
	no further action is required.
\end{enumerate}

If the algorithm did not reject the operation,
then it can be safely performed. The only step
left is to assign the value $0$ to any
$d \in \{a,b,c\}$ for which $m$ is undefined.

Carrying out the procedure sketched above requires a constant
numbers of reads and writes of the map $m$ whose efficiency
depends on the implementation of choice but can be safely
assumed as negligible.

\paragraph{Takeaway message}
As shown in the previous subsections,
well-stratification is an useful property
for NG families thus, being able to efficiently
check 
\begin{itemize}
	\item
		if a family is well-stratified or
	\item
		if an operation preserves well-stratification		
\end{itemize}
is of vital importance. 
In this subsection we described how
these two questions can be answered with a cost that is
linear in the number of assignments (\ie the size of meta-information)
and constant, respectively.
These results are strictly related to those described in Subsection~\ref{sec:ws-with-types} since the
dependency graph $G_n$ induced by a family $n$ is
a graph representation of the constraints on
$\Gamma$ derived by applying type inference to $n$.
Therefore, we can read the above results as complementing
the type system we introduced with an implementation.

\begin{figure}[H]
 \center
  \includegraphics[width=0.76\columnwidth]{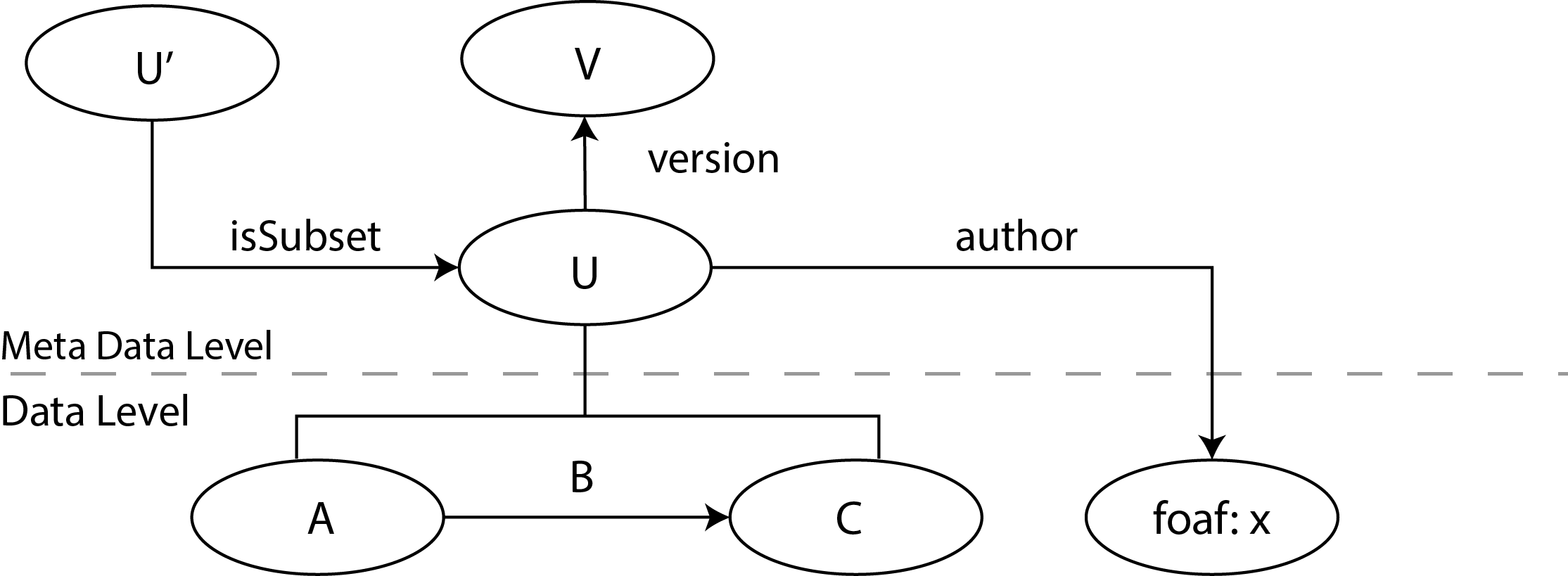}
  \caption{An example of well-stratified data.}
  \label{fig:dataex}
\end{figure}

\begin{figure}[H]
 \center
  \includegraphics[width=0.76\columnwidth]{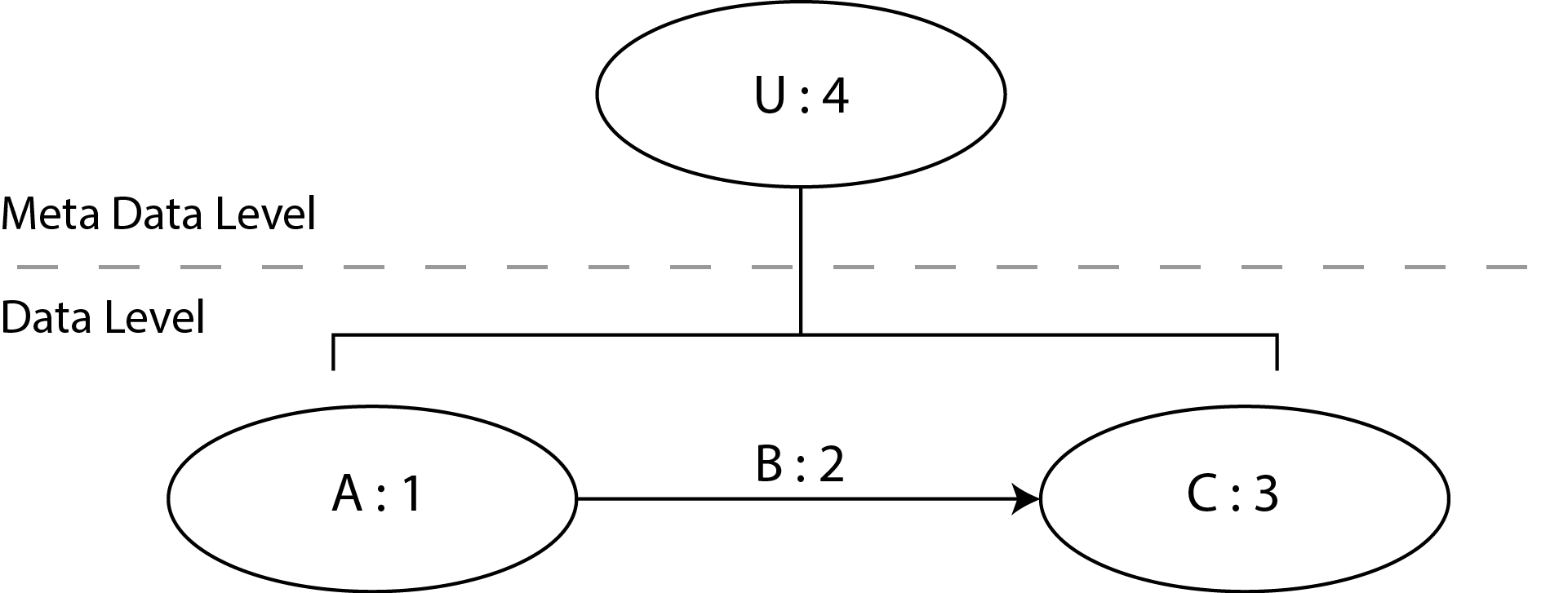}
  \caption{An example of meta-information levels assignment.}
  \label{fig:typeex}
\end{figure}

\section{Towards a Well-Stratified data language}
\label{sec:language}
In the previous sections we introduced the concept of well-stratified data, that is crucial to the practical realisation of data citation: as long as data is well-stratified resolving a data citation is always possible. 
Moreover we showed how assessing the well-stratification of a data set is linear with respect to the size of the data and can be done incrementally as new statements are added to the data set.
In this section we discuss how these notions may translate into practice.

Intuitively, for a data set to be well-stratified it means that it is always possible to draw a line separating information from meta-information while baring in mind that such separation is bound to be inherently relative to current datum. 
In order to illustrate this idea consider the simple example depicted in \Cref{fig:dataex}: the triple $(A, B, C)$ is the information, while its identifier $U$ and its related predicates represent the meta-information with respect to the datum considered \ie $(A,B,C)$. Intuitively $U$ is the IRI of the reified \emph{statement} that has as subject, object, and predicate respectively $A$, $C$, and $B$. Each IRI in the data set should belong to exactly one of these levels. Theoretically, there should be no upper bound to the number of meta-information levels since one might be interested in expressing statements over meta-information. An hypothetical reification, for instance, of $(U, \mathtt{version}, V)$ having as identifier $I$ would imply a further meta-information level containing only $I$.

\looseness=-1
This intuitive stratification can be seen as assigning to (meta-)information to levels that decreases as we move from meta-information to information. The lowest level contains all the data that data citation resolution should ultimately reach. Well-stratification not only ensures the existence of a lowest level but it also guarantees that wherever a reasoner starts unravelling the (meta-)information chain it will eventually reach said level. 

Indeed this is the approach described in \Cref{sec:coherent} and embodied by the function $\Gamma\colon V \rightharpoonup \mathbb{N}$ used by the type system introduced in \Cref{sec:ws-with-types}. From a concrete perspective this means that identifiers of a reified statement should therefore be assigned a higher level by $\Gamma$ than their subject, object, and predicate. This condition is met if, and only if, the data set is well-stratified. Finally, we remark that finding such levels can be done incrementally while reading the data set or on-the-fly as operations are performed on a well-stratified set; the cost of the former is linear in the size of the base whereas the cost of the second is linear on the number of operations.

The various data levels that can be identified this way can be considered as distinct data ``slices'' that, though linked, can be considered independent data sets and treated accordingly. 
For instance a large, multi-layered data set could be distributed as a whole, with all layers of meta-information or ``reduced'' to its sole data level, \ie, without meta-information.
Since well-stratification can be checked automatically  in a reasonably efficient way and allows the separation of meta-information from ground level information, it addresses the need for such a separation introduced by state of the art data citation methodologies \cite{silvello2015methodology}. Such a concept, however,  does not exist yet in the Semantic Web stack.
The RDF language in fact has several limitations that make well-stratification hard to realise, and here we pinpoint the most evident:
\begin{itemize}
\item Checking well-stratification implies, as shown in Section \ref{sec:ws-with-types}, the presence of a type-checking mechanism that does not exist in RDF.
\item There is a data-level usage of reification (for instance to express sentences like ``Bob says that Alice is kind'') that must not be confused with the labelling of triples for meta-information expression purpose we analysed so far, and RDF does not provide a way to discriminate them.
\item Assigning an identifier to a triple in RDF is not handy due to the bloated syntax of reification.
\end{itemize}
To overcome these limitations we strongly advocate the creation of a new language wherein the concept of well-stratification is a first-class citizen. More specifically, such a language should include in its specifications:
\begin{itemize}
\item a class for meta-information objects, allowing to explicitly state which triples are to be considered meta-information and which information;
\item a \emph{level} property that can be associated to any IRI, serving as explicit annotation of the information level the IRI belongs to;
\item a syntax for quad semantics, \ie, switching from a language of triples (like RDF) to a language of quadruples where triples are considered quadruples with a void fourth element;
\item a more restrictive semantics for the fourth element of the quadruple, allowing to discriminate between reification for meta-information annotation purpose from actual data level usage;
\item a \emph{type system} for data including well-stratified data.
\end{itemize}
Given such a language the actual information would be still expressed in the form of triples, allowing compatibility with the other levels of the Semantic Web, and the meta-information could be handled separately.

\section{Conclusions}
\label{sec:conclusions}
In this paper we briefly outlined a formalisation of data citation over linked data and showed how resolving authorship attribution, subset and version identification are computable in an efficient way as long as the considered data is well-stratified. 
Because of the relevance of this property, we explored how it can be expressed and verified by means of a type system and \adhoc algorithms.
We showed that checking whether a given NG family (which abstracts over the concrete form of data) is well-stratified requires linear time and proposed a constant time solution for checking if an operation performed by any reasoner preserves or breaks this property.

\looseness=-1
With respect to the problem of data citation, the expressive power of OWL and RDF is largely over-abundant and might be harmful since a misuse of their primitives might break the stratification of information and meta-information, thus making resolving citations an undecidable problem.
In our opinion, a more restricted language, designed specifically to grant the stratification of data should be taken into consideration to effectively enable problems such as data citation and provenance assessment to be resolved in practical time, allowing the creation of an effective data trust layer. Attaching meta-information to data published on the Web leveraging such a language might be, in our opinion, Linked Open Data's sixth star, like publishing versioned code is a fundamental quality requirement for Open Source software. The similarity between data meta-information handling and source code versioning is striking since they address similar problems: tracking who and how edited something, identifying subsets of the managed items, and allowing external application or documents to refer to a specific revision.
In our opinion this separation is also consistent with the present development of the Semantic Web stack: OWL itself, thought being a logical extension of RDFS, is not built on the top of RDFS but is rather a distinct language sharing concepts and primitives with RDFS. In a similar way a new language for data meta-information management could be built compatibly with RDF and the Linked Data philosophy without being RDF.

Finally, this work suggests a deeper connection between formal graph models and knowledge management problems encountered by the Digital Libraries and Semantic Web communities. In our opinion a more formal take on a broad range of non trivial knowledge management tasks and practises might provide relevant insights both on the application side and on the theoretical one as suggested by preliminary works in this directions like \eg \cite{wangH13,mansutti:eceast2014,mansutti:dais2014,miculan:csp2014}.

\end{multicols}
\end{document}